\documentclass[a4paper,runningheads]{llncs}
\usepackage[utf8]{inputenc}

\usepackage{graphicx,amssymb,amsmath,url}
\usepackage{amsfonts}
\usepackage[hidelinks]{hyperref}
\usepackage{thmtools,cleveref}
\usepackage{thm-restate}
\usepackage{cite}
\usepackage{complexity}
\usepackage{xpatch} 
\usepackage{subcaption}
\makeatletter
\newcommand\footnoteref[1]{\protected@xdef\@thefnmark{\ref{#1}}\@footnotemark}
\makeatother

\newcommand\rurl[1]{%
	\href{http://#1}{\nolinkurl{#1}}%
}

\spnewtheorem{observation}{Observation}{\bfseries}{\rmfamily}

\let\doendproof\endproof
\renewcommand\endproof{\hfill $\square$\doendproof}

\author{Philipp de Col 
\and Fabian Klute 
\and Martin N\"ollenburg } 
\title{Mixed Linear Layouts:\\ Complexity, Heuristics, and Experiments}
\institute{Algorithms and Complexity Group, TU Wien, Vienna, Austria
\email{philipp.decol@gmx.at,\{fklute,noellenburg\}@ac.tuwien.ac.at}}

\begin{document}

	\maketitle
		
	\begin{abstract}
		A $k$-page linear graph layout of a graph $G=(V,E)$ draws all vertices along a line $\ell$ and each edge in one of $k$ disjoint halfplanes called \emph{pages}, which are bounded by $\ell$.
		We consider two types of pages. In a \emph{stack page} no two edges should cross and in a \emph{queue page} no edge should be nested by another edge. 
		A crossing (nesting) in a stack (queue) page is called a \emph{conflict}. 
		The algorithmic problem is twofold and requires to compute (i) a vertex ordering and (ii) a page assignment of the edges such that the resulting layout is either conflict-free or conflict-minimal.
		While linear layouts with only stack or only queue pages are well-studied, mixed $s$-stack $q$-queue layouts for $s,q \ge 1$ have  received less attention.
		We show \NP-completeness results on the recognition problem of certain mixed linear layouts and present a new heuristic for minimizing conflicts. 
		In a computational experiment for the case $s, q = 1$ we show that the new heuristic is an improvement over previous heuristics for linear layouts.
	\end{abstract}
	
	\section{Introduction}\label{sec:introduction}
	
	Linear graph layouts, in particular book embeddings~\cite{ollmann73,BernhartK79} (also known as stack layouts) and queue layouts~\cite{HeathR92,Heath1992}, form a classic research topic in graph drawing with many applications beyond graph visualization as surveyed by Dujmović and Wood~\cite{DujmovicW04}.
	A $k$-page linear layout $\Gamma=(\prec, \mathcal P)$ of a graph $G=(V,E)$ consists of an order $\prec$ on the vertex set $V$ and a partition of $E$ into $k$ subsets $\mathcal P = \{P_1, \dots, P_k\}$ called \emph{pages}.
	Visually, we may represent $\Gamma$ by mapping all vertices of $V$ in the order $\prec$ onto a line $\ell$. %
	Each page can be represented by mapping all edges to semi-circles connecting their endpoints in a halfplane bounded by $\ell$.
	If a page $P$ is a \emph{stack page}, then no two edges in $P$ may cross, or at least the number of crossings should be minimized.
	More precisely, two edges $uv$, $wx$ in $P$ cross (assuming $u \prec v$, $w \prec x$, and $u \prec w$) if and only if their vertices are ordered as $u \prec w \prec v \prec x$.
	Conversely, if a page $P$ is a \emph{queue page}, then no two edges in $P$ may be nested, or at least the number of nestings should be minimized. 
	Here, an edge $wx$ is nested by an edge $uv$ if and only if their vertices are ordered as $u \prec w \prec x \prec v$ (under the same assumptions as above).

	Stack and queue layouts have mostly been studied for planar graphs with a focus on investigating the stack number (also called book thickness) and the queue number of graphs, which correspond to the minimum integer $k$, for which a graph admits a $k$-stack or $k$-queue layout. 	
	It is known that recognizing graphs with queue number 1 or with stack number 2 is \NP-complete~\cite{BernhartK79,HeathR92}.
	Further, it is known that every planar graph admits a 4-stack layout~\cite{Yannakakis89}, but it is open whether the stack number of planar graphs is actually 3.
	Due to their practical relevance, book drawings, i.e., stack layouts in which crossings are allowed, have also been investigated from a practical point of view. 
	Klawitter et al.~\cite{Klawitter2018} surveyed the literature and performed an experimental study on several state-of-the-art book drawing algorithms aiming to minimize the number of crossings in layouts on a fixed number of stack pages. 
	Conversely, for queue layouts it was a longstanding open question whether planar graphs have bounded queue number~\cite{Heath1992}; this was recently answered positively by Dujmović et al.~\cite{DujmovicJMMUW19}.
	
	Mixed layouts, which combine $s \ge 1$ stack pages and $q \ge 1$ queue pages, are studied less. 
	For an $s$-stack $q$-queue layout $\Gamma = (\prec, \mathcal P)$, the set of pages $\mathcal P$ is itself partitioned into the stack pages $\mathcal S = \{S_1, \dots, S_s\}$ and the queue pages $\mathcal Q = \{Q_1, \dots, Q_q\}$.
	Heath and Rosenberg~\cite{HeathR92} conjectured that every planar graph admits a 1-stack 1-queue layout, but this has been disproved recently by Pupyrev~\cite{p-mllpg-17}, who conjectured that instead every bipartite planar graph has a 1-stack 1-queue layout. 
	Pupyrev further provides a SAT-based online tool for testing the existence of an $s$-stack $q$-queue layout\footnote{\url{http://be.cs.arizona.edu}}.

	\smallskip
	\noindent\textit{Contributions.} 
	We first show two \NP-completeness results in Section~\ref{sec:hardness}.
	The first one shows that testing the existence of a 2-stack 1-queue layout is \NP-complete, and the other proves that an \NP-complete mixed layout recognition problem with fixed vertex order remains \NP-complete under addition of 
	stack or queue pages. 
	Next, we focus our attention on 
	1-stack 1-queue layouts 
	and propose, to the best of our knowledge, the first heuristic 
	targeted at minimizing conflicts in 1-stack 1-queue layouts, see Section~\ref{sec:heuristic}.
	In a computational experiment in Section~\ref{sec:experiments} we show that our heuristic achieves fewer conflicts compared to previous heuristics for stack layouts with a straightforward adaptation to mixed layouts.
	
	\smallskip \noindent {\emph{Statements whose proofs are located in the appendix are marked with $\star$.}}

	\section{Complexity}\label{sec:hardness}
	In this section we give new complexity results regarding mixed linear layouts.
	For Theorem~\ref{thm:1q2s} we first make some useful observations, see Appendix~\ref{app:hardness} for details.
	Let $ K_8 $ be the complete graph on eight vertices and
	$\Gamma =  (\prec,\{S_1,S_2\},\{Q\}) $ a 2-stack 1-queue layout of a $K_8$.
	Using exhaustive search\footnote{\label{ftn:code}Source code available at \url{https://github.com/pdecol/mixed-linear-layouts}}%
	we verified that in such a  2-stack 1-queue layout the three longest edges are in $ S_1 \cup S_2 $ and the edges between the first and third, and the sixth and eighth vertex in $ \prec $ are in $ Q $. Finally, for two $ K_8 $'s only the last and first vertex from each $ K_8 $ can interleave. 

	Let $ G_1 = (V_1,E_1) $ and $ G_2 = (V_2,E_2) $ be two distinct $K_8$'s. 
	A \emph{double-$ K_8 $} $ G $ is formed, by identifying two so-called \emph{shared} vertices $ u,v $ of $ G_1 $ and $ G_2 $ with each other
	and adding one more edge $ wz $ between so-called \emph{outer} vertices $ w \in V_1 $, $ z \in V_2 $ 
	such that neither $ w $ nor $ z $ is one of the shared vertices, see Fig.~\ref{fig:reduction}(a).
	For any 2-stack, 1-queue layout $\Gamma =  (\prec,\{S_1,S_2\},\{Q\}) $ of a double-$ K_8 $ we
	verified with exhaustive search that the shared vertices have to be the two middle vertices in $ \prec $
	and that the outer vertices are always the first and last vertices in the ordering of such a layout.

	\begin{restatable}[$ \star $]{lemma}{lemnoinsidevertex}
		\label{lem:no_inside_vertex}
		Let $ G = (V,E) $ be a double-$ K_8 $ with outer vertices $ w,z $
		and two additional vertices $ a,b \in V $ with an edge $ ab \in E $.
		In every 2-stack 1-queue layout $ \Gamma = (\prec,\{S_1,S_2\},\{Q\}) $ of $ G $,
		with $ w \prec a \prec z $ and $ b \prec w $ or $ z \prec b $, it holds that
		$ ab \in Q $ and $ a $ is between the first  or last three vertices in the double-$ K_8 $.
	\end{restatable}

	\begin{figure}[tbp]
		\centering 
		\includegraphics{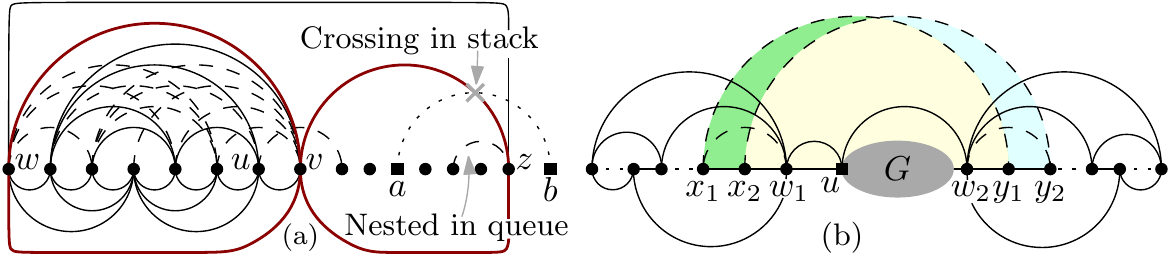}
		\caption{(a) A $ 2 $-stack $ 1 $-queue layout of a double-$ K_8 $. Only the left $ K_8 $ is drawn fully, dashed edges are in the queue page. (b) Sketch of the gadget for Theorem~\ref{thm:1q2s}.}
		\label{fig:reduction}
	\end{figure}
	
	\begin{restatable}[$ \star $]{corollary}{corinterleave}
		\label{cor:interleave}
		Let $ G_1 = (V_1,E_1) $ and $ G_2 = (V_2,E_2) $ be two double-$ K_8 $'s.
		In a $ 2 $-stack $ 1 $-queue layout $\Gamma$ of $ G_1 \cup G_2 $ with linear order $ \prec $,
		either $ u \prec v $ 
		or $ v \prec u $ for all $ u \in V_1 $, $ v \in V_2 $.
	\end{restatable}

	Let $ G = (V,E) $ be a graph consisting of two double-$ K_8 $'s $ G_1 = (V_1,E_1) $ and $ G_2 = (V_2,E_2) $.
	Let $ w_1 \in V_1 $ and $ w_2 \in V_2 $ be two outer vertices.
	Further, let $ x_1,x_2 \in V_1 $ and $ y_1,y_2 \in V_2 $ be four vertices such that
	$ x_1,x_2 $ are in the same $ K_8 $ as $ w_1 $,
	$ w_2 $ respectively for $ y_1,y_2 $,
	and none of them is a shared vertex.
	Finally, add a vertex $ u $ to $ V $ and
	the edges $ x_1y_1 $, $ x_2y_2 $, $ w_1u $, and $ w_2u $ to $ E $,
	see Fig.~\ref{fig:reduction}(b).
	
	\begin{restatable}[$ \star $]{lemma}{lempositioning}\label{lem:positioning}
		Let $ G = (V,E) $ be the graph constructed as above.
		Then in any $ 2 $-stack $ 1 $-queue layout $ \Gamma = (\prec,\{S_1,S_2\},\{Q\}) $ of $ G $ we find, w.l.o.g.,
		$ w_1 \prec u \prec w_2 $ and $ w_1u, w_2u \in S_1 \cup S_2 $.
	\end{restatable}

	\begin{theorem}\label{thm:1q2s}
		Let $ G = (V,E) $ be a simple undirected graph. It is \NP-complete to decide if $ G $ admits a 2-stack 1-queue layout.
	\end{theorem}	
	\begin{proof}
		 The problem is clearly in \NP. We show the result by a reduction from the problem of deciding the existence of a 2-stack layout, which is \NP-complete and equivalent to decide whether a graph is subhamiltonian~\cite{BernhartK79,ChungLR87}. 
		Let $ G' = (V',E') $ be a graph constructed as in Lemma~\ref{lem:positioning}.
		Identify the special vertex $ u $ in $ V' $ with any vertex in $ G $ and add the rest of $ G $ to $ G' $.
		Clearly, if $ G $ has a $ 2 $-stack layout, we can construct a 2-stack 1-queue layout of $ G' $ as sketched in Fig.~\ref{fig:reduction}(b).
		Conversely, let $\Gamma = (\prec,\{S_1,S_2\},\{Q\}) $ be a $ 2 $-stack $ 1 $-queue layout of  $ G' $.
		As for $ u $ in Lemma~\ref{lem:positioning}, we find that $ w_1 \prec v \prec w_2 $ for every neighbor $ v \in V $ of $ u $.
		By induction we find for all $ v' \in V $ that $ w_1 \prec v' \prec w_2 $.
		Hence all edges in $ G $ are nested by $ x_1y_1 $ and $ x_2y_2 $, which are both in $ Q $.
		It follows that $ G $ has a $ 2 $-stack layout.
	\end{proof}

	Our second complexity result shows that adding stack or queue pages to the specification of an already \NP-complete mixed linear layout problem with given vertex order $\prec$ remains \NP-complete. Note that for $ s = 4 $ and $ q = 0 $ the problem of deciding if the edges can be assigned to the pages even when the vertex order is fixed is known to be \NP-complete~\cite{unger88}.

	\begin{restatable}[$ \star $]{theorem}{thmaddpages}
		\label{thm:addqadds}
		Let $ G = (V,E) $ be a simple undirected graph and $\prec$ a fixed  order of $V$. If it is \NP-complete to decide if $ G $ admits an $ s $-stack $ q $-queue layout respecting $\prec$, then it is also \NP-complete to decide if $ G $ admits  (i) an $ s $-stack $ (q + 1) $-queue layout or (ii) an  $ (s +1) $-stack $ q $-queue layout respecting $\prec$, respectively.
	\end{restatable}

	\section{Heuristic Algorithm}\label{sec:heuristic}
	Most heuristics for minimizing crossings in book drawings work in two steps~\cite{Klawitter2018}.
	First compute a vertex order and then a page assignment. 
	We propose a new page assignment heuristic, specifically tailored to mixed linear layouts. 
	To the best of our knowledge, this is the first heuristic for minimizing crossings and nestings in mixed linear layouts. %
	It uses stack and queue data structures for keeping track of conflicts and estimating possible future conflicts.
	The design allows us to consider the assigned and unassigned edges at the same time while efficiently processing them to run in $O(m^2)$ time for a graph with $m$ edges.

In the following, we describe how the algorithm works for a 1-stack 1-queue layout. 
We note that it is straight forward to adapt our approach to $ s $-stack, $ q $-queue layouts for arbitrary $ s $ and $ q $.
A stack (queue) can be used to validate that a given page has no crossings (nestings) by visiting the vertices in their linear order and inserting (removing) each edge when the left (right) end-vertex is visited, respectively~\cite{Heath1992}. 
We can use a similar strategy for our heuristic.
Here it is allowed to remove edges even if they are not on top of the stack or in front of the queue, but of course this might produce conflicts. 
Let $ S $ be a stack and $ Q $ a queue.
We additionally keep two counters for each edge $ e $, the so called \emph{crossing counter} $ c(e) $ and the \emph{nesting counter} $ n(e) $.
The vertices are processed from left to right in a given vertex order. 
For the current vertex $ u $ we insert all edges $ e = uv $ into $ S $ and $ Q $.
If there are multiple edges $ uv $, we add them according to their length to $ S $ and $ Q $.
For $ S $ we sort them from long to short, for $ Q $ from short to long.
Once the second vertex $v$ of an edge $e=uv$ is visited, we remove $e$ from $ S $ and $ Q $, and decide to which page we assign it. 
Let $ s_e $ be the number of edges on top of $ e $ in $ S $ and $ q_e $ the number of edges  in front of $ e $ in $ Q $.
Then we assign $ e $ to the stack page if $ c(e) + 0.5s_e \leq n(e) + 0.5q_e $ and update $ c(e') $ for each edge $ e' $ on top of $ e $ in $ S $.
Otherwise we assign $ e $ to the queue page and update $n(e')$ for each edge $e'$ in front of $ e $ in $ Q $. 
Intuitively we estimate for each edge $ e $ how many conflicts this edge produces for edges $ e' $ to be processed later.
The advantage of this estimation
is that we potentially assign the edge to a page that adds more conflicts now, but might create fewer conflicts in the future.

    \section{Experiments}\label{sec:experiments}
	We denote our algorithm as \emph{stack-queue} heuristic and  compare it with the two page assignment heuristics \emph{eLen}~\cite{Cimikowski2002} and \emph{ceilFloor}~\cite{Satsangi2013} that are commonly used 
	with book drawings~\cite{Klawitter2018}.\footnoteref{ftn:code}
	Both 
	can be adapted to mixed layouts, while other book drawing heuristics try to explicitly partition the edges into planar sets, which is obviously not suitable for queue pages in mixed layouts. 
	Both process the edges by decreasing length and greedily assign each edge to the page where it causes fewer conflicts at the time of insertion. 
	In case of ties, a stack page is preferred over a queue page.
	The difference is that \emph{eLen} computes the length based on the linear vertex order and \emph{ceilFloor} based on the corresponding cyclic vertex order as follows.
Given a vertex order $ 1 \prec 2 \prec \ldots \prec n $, \emph{eLen} considers the edge $ (1, n) $ first and the edges $ (i, i + 1) $  last. 
In \emph{ceilFloor} the length of an edge $uv$ is defined as $ min(|u - v|, n - |u - v|) $. 
All three heuristics run in $O(m^2)$. %

The goal of our experiment is to explore the performance differences in terms of the number of conflicts per edge of the adapted book drawing algorithms compared to our new heuristic. 
We thus measure 
the resulting number of conflicts per edge for all three algorithms, as well as record for each instance the algorithm with the fewest number of conflicts.
We first tested the algorithms on the complete graphs with up to $ 50 $ vertices.
Furthermore, we generated 500 random graphs for each number of vertices in $\{25, 50, \dots, 400\}$ from different \emph{sparse} graph classes, see Fig.~\ref{fig:winner_400_vertices}, since it is known that both, stack and queue page,
can contain at most $2n-3$ conflict-free edges~\cite{BernhartK79,HeathR92}.
All experiments ran on a Linux cluster (Ubuntu 16.04.6 LTS), where each node has two Intel Xeon E5540 (2.53 GHz Quad Core) processors and 24GB~RAM. %
        The running time of one run of each algorithm was relatively low,
    	taking less than one second on average and at most 2.5 seconds for the denser random graphs with $ 400 $ vertices.

    For the complete graphs 
    \emph{stack-queue} was the best page assignment heuristic producing about $2/3$ of the conflicts of \emph{eLen} and \emph{ceilFloor}.
	For other graphs, the vertex order has a strong effect on the results. In our experiments 
	we first compared three state-of-the-art vertex order heuristics (breadth-first search (\emph{rbfs}~\cite{Satsangi2013}), depth-first search (\emph{AVSDF}~\cite{He2004}) and connectivity (\emph{conGreedy}~\cite{Klawitter2018})) before applying the page assignment heuristics. It turned out that for all of them the same vertex order heuristic performed best on the same graph class, so that all experiments could be run without bias on exactly the same input order.

	\smallskip
\noindent\textit{Benchmark graphs.}   
We first generated random (not necessarily planar) graphs of $n$ vertices and either $m=3n$ or $m=6n$ edges. 
The graphs were created by drawing uniformly at random the required number of edges, discarding disconnected graphs. The best vertex order heuristic was \emph{conGreedy}. %

Since 1-stack 1-queue layouts are especially interesting for planar and planar bipartite graphs~\cite{p-mllpg-17}, we generated random planar 
and maximal planar bipartite graphs. The planar graphs are generated as 
Delaunay triangulations of $ n $ random points in the plane.
Since every planar bipartite graph has a 2-stack embedding~\cite{DBLP:journals/dcg/FraysseixMP95}, we randomly generated a vertex order of alternating vertices from both vertex sets to ensure that a Hamiltonian path exists. We then randomly selected two vertices of the two sets and added the edge to the graph if it was possible to do so without a crossing. 
We repeated the process of randomly selecting the vertices until the maximum number of $ 2n - 4 $ edges had been reached.
The vertex order for the Delaunay triangulations was computed by \emph{rbfs} and for maximal planar bipartite graphs by \emph{AVSDF}.
As it turned out that \emph{stack-queue} did not perform as well as \emph{ceilFloor} and \emph{eLen} on the Delaunay triangulations, which is in contrast to the random and planar bipartite graphs, we wondered whether the presence of many triangles 
might be the reason.
Hence, we considered two graph classes with many triangles and the same maximal edge densities as planar and planar bipartite graphs, respectively, namely planar 3- and 2-trees. For 2-trees the best vertex order was computed by \emph{AVSDF} and for 3-trees by \emph{conGreedy}.
	
	\begin{figure}[tb]
		\centering
		\begin{subfigure}{.33\textwidth}
			\centering
			\includegraphics[width=1\textwidth]{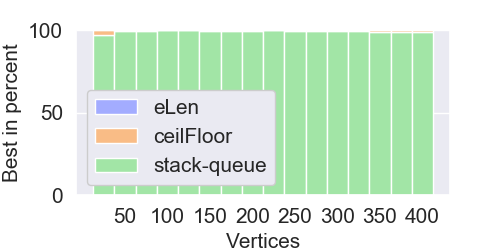}
			\caption{Random $m=3n$}
		\end{subfigure}%
		\begin{subfigure}{.33\textwidth}
			\centering
			\includegraphics[width=1\textwidth]{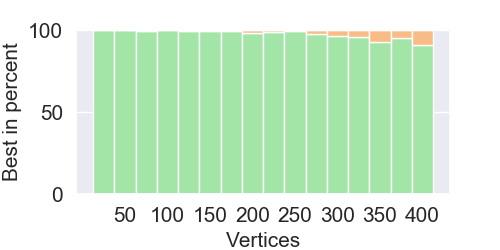}
			\caption{Random $m=6n$}
		\end{subfigure}%
		\begin{subfigure}{.33\textwidth}
			\centering
			\includegraphics[width=1\textwidth]{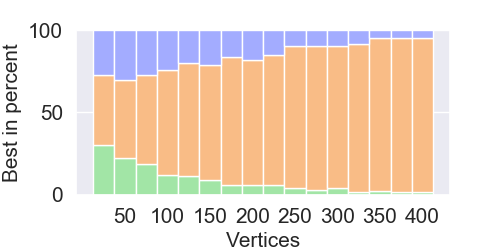}
			\caption{Delaunay triangulations}
		\end{subfigure}\\ %
		\begin{subfigure}{.33\textwidth}
			\centering
			\includegraphics[width=1\textwidth]{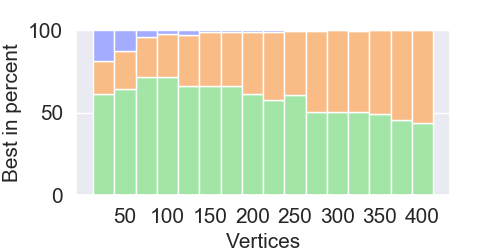}
			\caption{Planar bipartite graphs}		
		\end{subfigure}%
		\begin{subfigure}{.33\textwidth}
			\centering
			\includegraphics[width=1\textwidth]{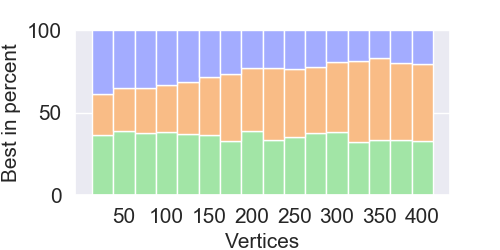}
			\caption{2-trees}		
		\end{subfigure}%
		\begin{subfigure}{.33\textwidth}
			\centering
			\includegraphics[width=1\textwidth]{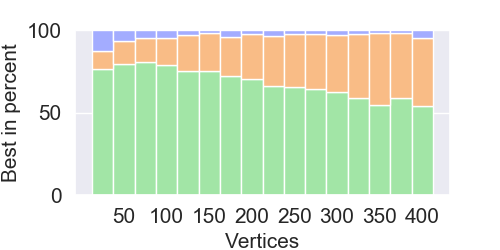}
			\caption{3-trees}		
		\end{subfigure}%
		\caption{How many times each algorithm obtained the fewest conflicts in percent.}
		\label{fig:winner_400_vertices}
	\end{figure}		
	
		\smallskip
	\noindent\textit{Results.}
	Aggregated results of our experiments are plotted in Fig.~\ref{fig:winner_400_vertices} and~\ref{fig:cpe_400_vertices}. 	%
	For random graphs
	\emph{stack-queue} performs best among the three heuristics for almost all instances, even though the difference to \emph{ceilFloor} in conflicts per edge is small. 
For Delaunay triangulations, \emph{stack-queue} performs  best for small graphs ($ n \leq 25 $), but for the larger instances \emph{ceilFloor} computes  better solutions for the majority of instances. In terms of conflicts per edge, however, all three algorithms are quite close together. In the case of planar bipartite graphs, \emph{stack-queue} is the best algorithm for up to $ 300 $ vertices. Afterwards \emph{ceilFloor} performs slightly better, but in both cases, again, the difference in the number of conflicts is small. For 2-trees, the results are more or less evenly split among all three algorithms. %
	Yet, for 3-trees, \emph{stack-queue} computes the best solutions for 70--80\% of the instances with up to $100$ vertices. The differences in the number of conflicts per edge is also more noticeable. For larger instances \emph{ceilFloor} catches up with \emph{stack-queue}.

		\smallskip
\noindent\textit{Discussion.}
The results of our experiments showed that the proposed \emph{stack-queue} heuristic beats or competes with previously existing and suitably adapted page assignment heuristics for book drawings on most of the tested benchmark graph classes with the exception of Delaunay triangulations, where \emph{ceilFloor} performed best. Since the running time of all three algorithms is $O(m^2)$ this does make \emph{stack-queue} a suitable method for computing 1-stack 1-queue layouts.

	\section{Conclusion}
	We believe it is possible to adapt our technique from Theorem~\ref{thm:1q2s} for $ s > 2 $ and $ q = 1 $. The biggest obstacle is to find such rigid structures as the double-$ K_8 $. In the algorithmic direction it could be interesting to investigate specialized heuristics for finding vertex orders in the queue- and mixed layout case. Whether every planar bipartite graph admits a $ 1 $-stack $ 1 $-queue layout remains open.

	\bibliographystyle{splncs04}
	\bibliography{paper}

\clearpage
	\appendix
	\section{Omitted Proofs from Section~\ref{sec:hardness}}\label{app:hardness}
	The following observations follow from exhaustively searching over all possible orders of the considered graphs.
	
	\begin{observation}\label{obs:K8}
		Let $ \Gamma = (\prec,\{S_1,S_2\},\{Q\}) $ be a 2-stack 1-queue layout of~$ K_8 $ and let the vertices $\{v_1, \dots, v_8\}$ be ordered by $\prec$. Then 		$ v_1v_8, v_1v_7, v_2v_8 \in S_1 \cup S_2 $,
		$ v_1v_7 $ and $ v_2v_8 $ are not in the same stack,
		and $ v_1v_3, v_6v_8 \in Q $.
	\end{observation}
	
	\begin{observation}\label{obs:K8andK8}
		Let $ G $ and $ H $ be two distinct $K_8$'s. Of the vertices of $ G $  and $ H $ only the first and last can interleave in any order of $ 2 $-stack, $ 1 $-queue layout of $ G \cup H $.
	\end{observation}

	\begin{observation}\label{obs:K8K8}
		Let $ G = (V_1\cup V_2,E) $ be a double-$ K_8 $ with shared vertices $ u,v $, outer vertices $ w,z $ and
		$\Gamma =  (\prec,\{S_1,S_2\},\{Q\}) $ a 2-stack 1-queue layout of $ G $. We find w.l.o.g.
		that $ w \prec v_1^1\prec \ldots \prec v_5^1 \prec u \prec v \prec v_1^2 \prec \ldots \prec v_5^2 \prec z$,
		where $ v_i^1 \in V_1 $ and $ v_i^2 \in V_2 $ for all $ i \in \{1,\ldots,5\} $.
	\end{observation}
	
	\lemnoinsidevertex*
	\begin{proof}
		Let $ u,v \in V $ be the shared vertices of the double-$ K_8 $ in $ G $.
		Assume there was a $ 2 $-stack $ 1 $-queue layout $ \Gamma $ 
		with $ w \prec a \prec z $ and $ b \prec w $ or $ z \prec b $ and
		$ ab \in S_1\cup S_2 $.
		From combining Observations~\ref{obs:K8} and~\ref{obs:K8K8} we obtain that the edges 
		$ wv $, $ vz $, $ wu $, and $ uz $ are all in $ S_1 \cup S_2 $.
		Furthermore if $ wv $ and $ vz $ are in $ S_1 $ it follows that $ wu $ and $ uz $ are in $ S_2 $.
		Hence, if $ ab \in S_1\cup S_2 $, it has to cross one of these four stack edges.
		It remains to show that $ a $ is between the first three or last three vertices of the double-$ K_8 $ in $ \prec $.
		Assume this was not the case, then the edge $ ab $ nests a queue edge by Observation~\ref{obs:K8}, see Fig.~\ref{fig:reduction}(a).
	\end{proof}

	\corinterleave*
	\begin{proof}
		Observe that it is not possible to interleave one vertex in $ V_1 $ with one in $ V_2 $ and vice versa since 
		the edges connecting outermost and shared vertices are all in the stacks.
		Any other possibility directly contradicts Observation~\ref{obs:K8andK8}.
	\end{proof}

	\lempositioning*
	\begin{proof}
		From Corollary~\ref{cor:interleave} we already know that the vertices of the double-$ K_8 $'s can only come sequentially after each other in $ \prec $.		
		With the same argument as in Lemma~\ref{lem:no_inside_vertex} 
		we immediately see that $ x_1,x_2 $ and $ y_1,y_2 $ have to be the two vertices next to the outer vertices $ w_1 $ and $ w_2 $
		in their respective double-$ K_8 $'s.
		Additionally we see that the edges $ x_1y_1 $ and $ x_2y_2 $ have to be in $ Q $ from Observations~\ref{obs:K8} and~\ref{obs:K8K8}.
		Thus it follows that they have to cross each other.
		In the following assume, w.l.o.g., that $ x_1\prec x_2 $ and $ y_1 \prec y_2 $.
		Finally, with Lemma~\ref{lem:no_inside_vertex}, we know that the only possible placements of $ u $ are $ x_1 \prec u \prec x_2 $, $ x_2 \prec u \prec w_1 $, $ w_1 \prec u \prec w_2 $, $ w_2 \prec u \prec y_2 $, and $ y_1 \prec u \prec y_2 $.

		Now assume $ u $ was not between $ w_1 $ and $ w_2 $, then, again with Observations~\ref{obs:K8} and~\ref{obs:K8K8},
		we know that either $ uw_1 $ or $ uw_2 $ has to be in $ Q $.
		Without loss of generality, let $ u $ be between $ x_1,x_2 $, and $ w_1 $,
		then $ uw_2 \in Q $, but also $ uw_2 $ is nested by $ x_1y_1 $, a contradiction.
	\end{proof}

	\begin{figure}[tbp]
		\centering
		\includegraphics{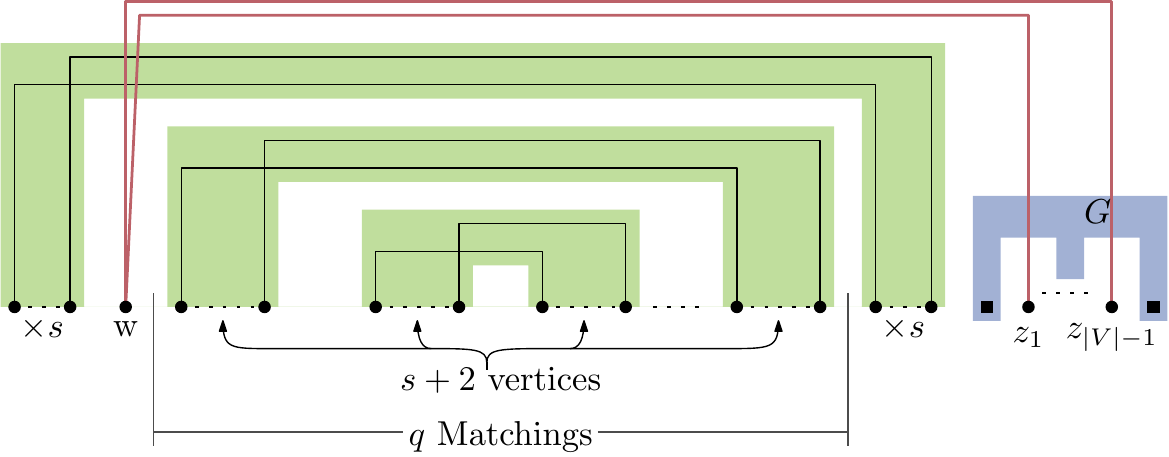}
		\caption{Illustration of the reduction in Theorem~\ref{thm:addqadds}}
		\label{fig:fixed_order_extra_stack_page}
	\end{figure}

	\thmaddpages*
	\begin{proof}
		Claim (i) is straightforward. 
		Assume it is \NP-complete to determine for a graph $ G $ if it has a $ s $-stack, $ q $-queue layout with fixed vertex order $ \prec $. 
		We create a new graph $ G' $ by adding a matching $ M $ on $ 2(s+1) $ new vertices to $ G $. 
		Let $ uv $ be an edge of this matching, then we place $ u $ before the vertices in $ V $ and $ v $ after all vertices in $ V $, i.e., $ uv $ nests all edges of $ G $. 
		Finally we order the vertices of $ M $ such that the edges all pairwise intersect. 
		It is clear that all edges of $ M $ can fit on one queue page. 
		Further, since there are $ s + 1 $ many, one such edge needs to be on a queue page, but it cannot be in any queue page together with an edge in $ E $, since it would nest all edges in $ E $. 
		The correctness follows.
		
		For Claim (ii) the reduction is slightly more complicated. 
		Let $ G = (V,E) $ be a graph and $ \prec $ a linear ordering of $ V $, such that it is \NP-complete to determine if an $ s $-stack $ q $-queue layout  respecting $ \prec $ exists. 
		First we add a star $ Z $ on $ |V| $ new vertices to $ G $. Let $ w $ be the center of this star. 
		We place it to the left of the first vertex of $ V $ in $ \prec $. 
		Now we add one of the $ |V| - 1 $ leaves of $ Z $ in between each pair of consecutive vertices in $ V $. %
		Observe that each edge $ uv \in E $ intersects at least one edge $ wz $ with $ z $ being a leaf of $ Z $.
		
		We further add $ q $  matchings $ M_i $, $ i = 1, \ldots, q $ on $ 2(s + 2) $ vertices. 
		We place each $ M_i $ between $ w $ and the vertices of $ V $ in $ \prec $, i.e., all edges in all $ M_i $ are nested by all edges of $ Z $. 
		We further place the vertices of $ M_i $ such that the edges of $ M_i $ nest all edges of $ M_j $, $ j = 1, \ldots, i - 1 $ and sort them such that the edges pairwise intersect. 
		Finally we add one more matching $ M $ on $ s $ vertices. For $ M $ we place one bipartition to the left of $ w $ in $ \prec $ and the other between the right-most vertex of $ M_q $ and the left-most vertex in $ V $. 
		See Figure~\ref{fig:fixed_order_extra_stack_page} for an illustration. 

		For the correctness we first observe that at least one edge of each $ M_i $ has to go to a queue. 
		It follows that for each queue we find at least one edge of one $ M_i $ and that no edge between vertices of $ M $ can be on a queue page. 
		From this we get that on each stack page we find exactly one edge of $ M $. 
		Finally this means that all edges of $ Z $ must be in one stack, and hence no edge in $ E $ can be in this stack. 
		It follows that if $ G $ has an $ s $-stack $ q $-queue layout then we always find a layout of the augmented graph on $ s + 1 $ stacks and $ q $ queues.
		Conversely, if we find such a layout of the augmented graph by construction all edges in $ E $ are on at most $ s $ stacks and at most $ q $ queues.
	\end{proof}
	
	\section{Additional Plots for Section~\ref{sec:experiments}}
	
	Figure~\ref{fig:cpe_400_vertices} shows the number of conflicts per edge for our experiments with the heuristically optimized vertex orders. Additionally, Figure~\ref{fig:cpe_random_vertex_order} shows the number of conflicts per edge for graphs with up to 100 vertices based on a random vertex order instead. Comparing the two figures, it is no surprise that with a random vertex order the conflict numbers increase. Further, for the Delaunay graphs, \emph{stack-queue}  achieves the lowest conflicts per edge with random vertex order, unlike for the optimized order, where \emph{ceilFloor} performs best. Overall, we observe that the error bands in the plots are more clearly separated for the random vertex order and that \emph{stack-queue} achieves on average one crossing per edge less than \emph{ceilFloor} on all six graph classes for $n=100$. This is not the case for the error bands in Fig.~\ref{fig:cpe_400_vertices} using  optimized vertex orders.

	\begin{figure}[tbp]
	\centering
	\begin{subfigure}{.5\textwidth}
		\centering
		\includegraphics[width=1\textwidth]{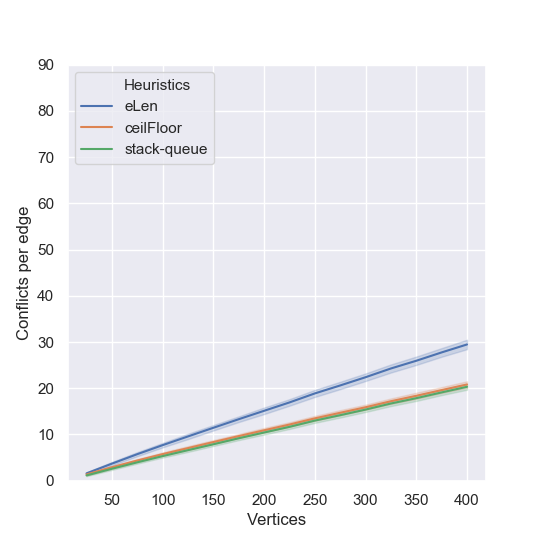}
		\caption{Random $m=3n$}
	\end{subfigure}%
	\begin{subfigure}{.5\textwidth}
		\centering
		\includegraphics[width=1\textwidth]{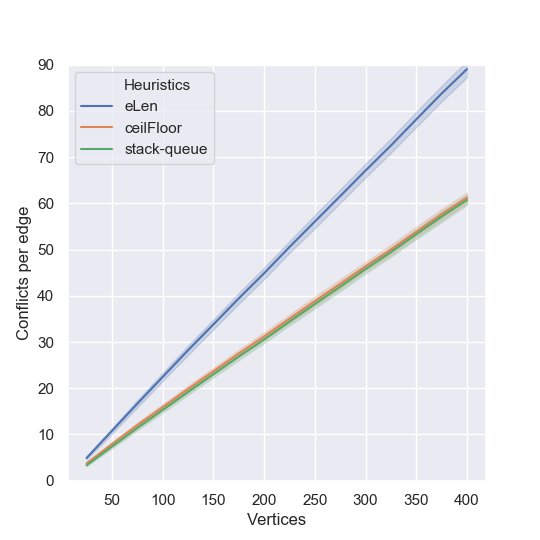}
		\caption{Random $m=6n$}
	\end{subfigure}
	\begin{subfigure}{.5\textwidth}
		\centering
		\includegraphics[width=1\textwidth]{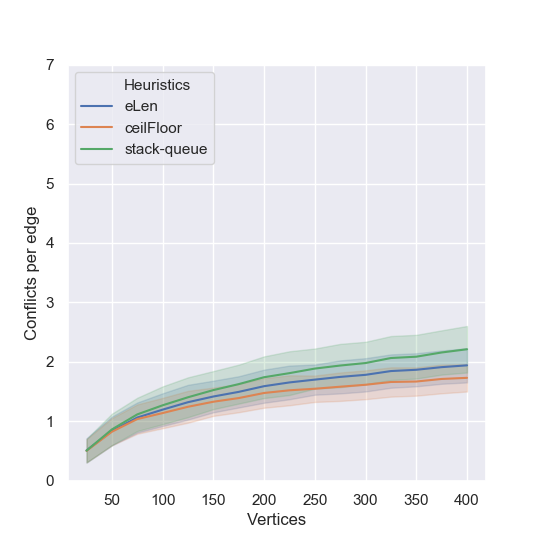}
		\caption{Delaunay triangulations}		
	\end{subfigure}%
	\begin{subfigure}{.5\textwidth}
		\centering
		\includegraphics[width=1\textwidth]{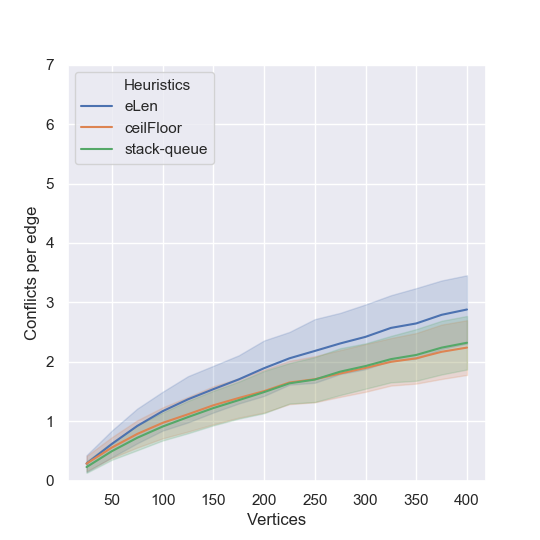}
		\caption{Planar bipartite graphs}		
	\end{subfigure}
	\begin{subfigure}{.5\textwidth}
		\centering
		\includegraphics[width=1\textwidth]{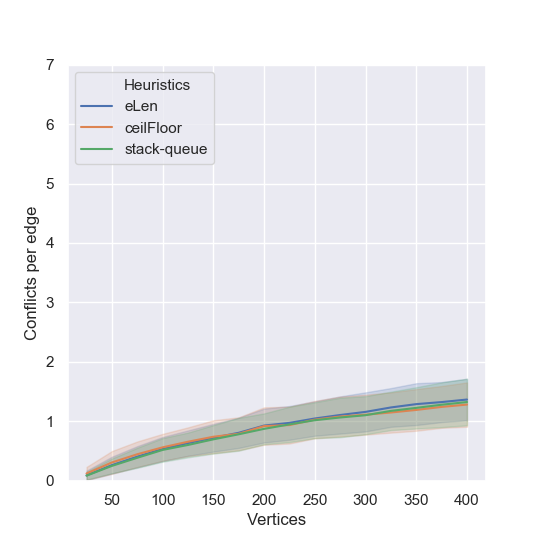}
		\caption{2-trees}		
	\end{subfigure}%
	\begin{subfigure}{.5\textwidth}
		\centering
		\includegraphics[width=1\textwidth]{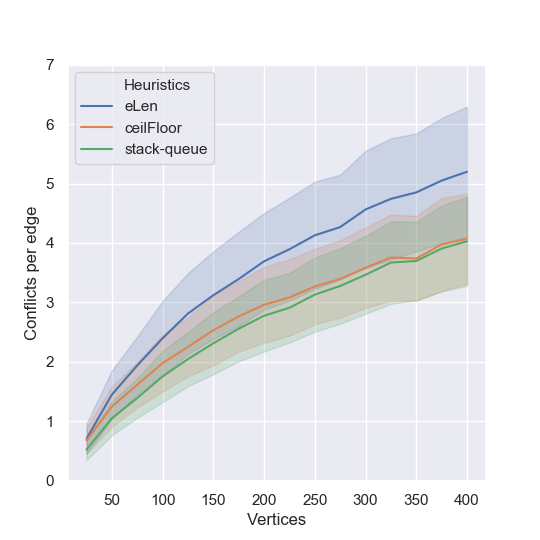}
		\caption{3-trees}
	\end{subfigure}
	\caption{Number of conflicts per edge for the three heuristics and six benchmark graph classes with heuristically optimized vertex orders.
	}\label{fig:cpe_400_vertices}
\end{figure}

\begin{figure}[tbp]
	\centering
	\begin{subfigure}{.5\textwidth}
    	\centering
    	\includegraphics[width=1\textwidth]{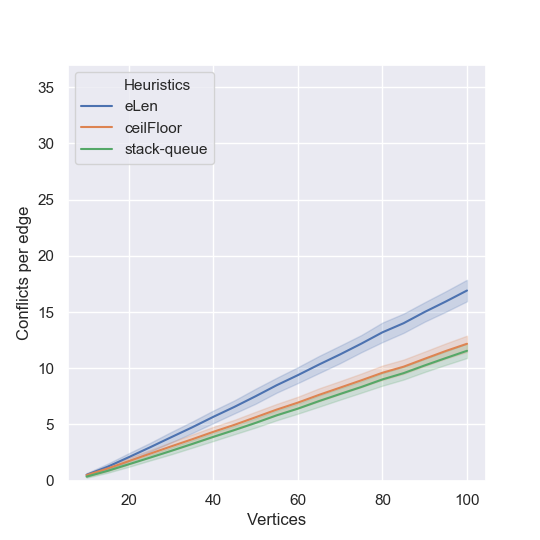}
		\caption{Random $m=3n$}
	\end{subfigure}%
	\begin{subfigure}{.5\textwidth}
    	\centering
    	\includegraphics[width=1\textwidth]{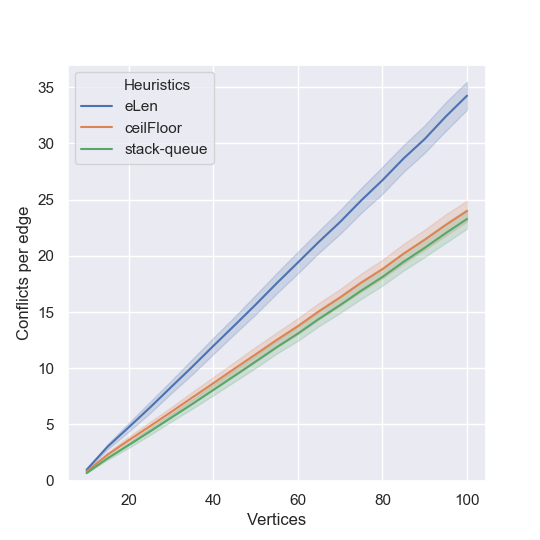}
		\caption{Random $m=6n$}
	\end{subfigure}
	\begin{subfigure}{.5\textwidth}
    	\centering
    	\includegraphics[width=1\textwidth]{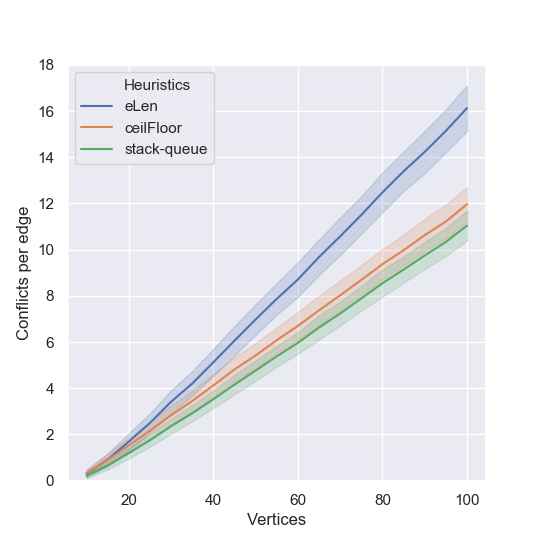}
		\caption{Delaunay triangulations}		
	\end{subfigure}%
	\begin{subfigure}{.5\textwidth}
    	\centering
    	\includegraphics[width=1\textwidth]{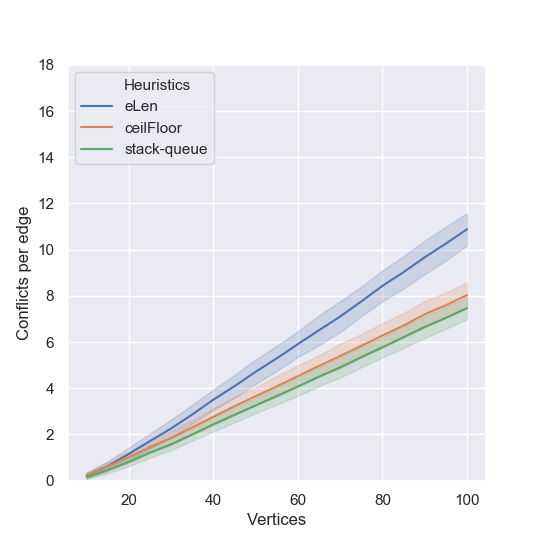}
		\caption{Planar bipartite graphs}		
	\end{subfigure}
	\begin{subfigure}{.5\textwidth}
    	\centering
    	\includegraphics[width=1\textwidth]{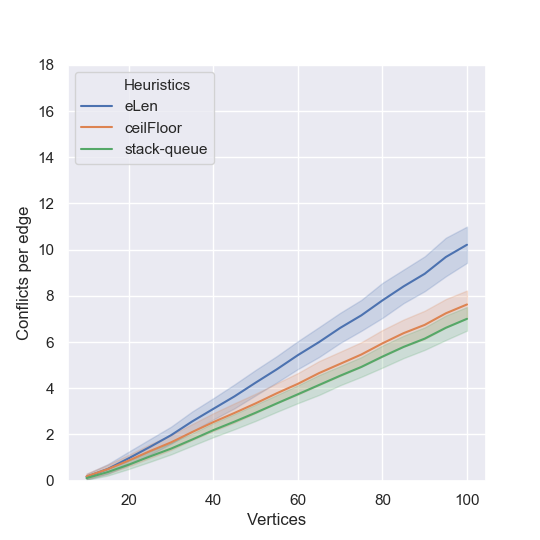}
		\caption{2-trees}		
	\end{subfigure}%
	\begin{subfigure}{.5\textwidth}
    	\centering
    	\includegraphics[width=1\textwidth]{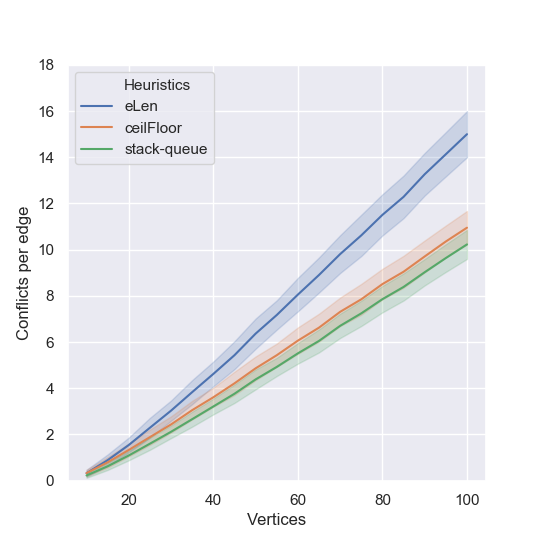}
		\caption{3-trees}
	\end{subfigure}
	\caption{Number of conflicts per edge for the three heuristics and six benchmark graph classes using random initial vertex orders.
}\label{fig:cpe_random_vertex_order}
	\end{figure}

\end{document}